\documentclass[preprint,12pt]{elsarticle}

\usepackage{amssymb}
\usepackage{amsthm}
\usepackage{amsfonts}
\usepackage{todonotes} % For todo notes
\presetkeys{todonotes}{inline}{}

\newtheorem{theorem}{Theorem}
\newtheorem{lemma}{Lemma}

\newtheorem{definition}{Definition}
\newtheorem{proposition}{Proposition}

\newcommand{\dist}{\textsf{dist}}

\newcommand{\ca}{\texttt{CA}}

\newcommand{\NP}{\textsf{NP}}
\newcommand{\coNP}{\textsf{coNP}}
\newcommand{\FPT}{\textsf{FPT}}

\newcommand{\yes}{\textsc{Yes}}
\newcommand{\no}{\textsc{No}}

\newcommand{\calA}{\mathcal{A}}
\newcommand{\calB}{\mathcal{B}}

\newcommand{\calO}{\mathcal{O}}
\newcommand{\calQ}{\mathcal{Q}}
\newcommand{\calR}{\mathcal{R}}

\newcommand{\defproblem}[3]{
  \vspace{1mm}
\noindent\fbox{
  \begin{minipage}{0.96\textwidth}
  \begin{tabular*}{\textwidth}{@{\extracolsep{\fill}}lr} #1 \\ \end{tabular*}
  {\bf{Input:}} #2  \\
  {\bf{Question:}} #3
  \end{minipage}
  }
  \vspace{1mm}
}
\begin{document}

\begin{frontmatter}

%% Title, authors and addresses

%% use the tnoteref command within \title for footnotes;
%% use the tnotetext command for theassociated footnote;
%% use the fnref command within \author or \address for footnotes;
%% use the fntext command for theassociated footnote;
%% use the corref command within \author for corresponding author footnotes;
%% use the cortext command for theassociated footnote;
%% use the ead command for the email address,
%% and the form \ead[url] for the home page:
%% \title{Title\tnoteref{label1}}
%% \tnotetext[label1]{}
%% \author{Name\corref{cor1}\fnref{label2}}
%% \ead{email address}
%% \ead[url]{home page}
%% \fntext[label2]{}
%% \cortext[cor1]{}
%% \affiliation{organization={},
%%             addressline={},
%%             city={},
%%             postcode={},
%%             state={},
%%             country={}}
%% \fntext[label3]{}

\title{Sparsification Lower Bound for Linear Spanners in Directed Graphs}

%% use optional labels to link authors explicitly to addresses:
%% \author[label1,label2]{}
%% \affiliation[label1]{organization={},
%%             addressline={},
%%             city={},
%%             postcode={},
%%             state={},
%%             country={}}
%%
%% \affiliation[label2]{organization={},
%%             addressline={},
%%             city={},
%%             postcode={},
%%             state={},
%%             country={}}

\author{Prafullkumar Tale \fnref{label2}}
\ead{prafullkumar.tale@cispa.saarland}
\affiliation{organization={CISPA Helmholtz Center for Information Security},%Department and Organization
            city={Saarbr$\ddot{u}$cken},
            postcode={66123}, 
            state={Saarland},
            country={Germany}}

\fntext[label2]{This research is a part of a project that has received funding from the European Research Council (ERC) under the European Union's Horizon $2020$ research and innovation programme under grant agreement SYSTEMATICGRAPH (No. $725978$).}

\begin{abstract}
%% Text of abstract
For $\alpha \ge 1$, $\beta \ge 0$,  and a graph $G$, a spanning subgraph $H$ of $G$ is said to be an $(\alpha, \beta)$-spanner if $\dist(u, v, H) \le \alpha \cdot \dist(u, v, G) + \beta$ holds for any pair of vertices $u$ and $v$.
These type of spanners, called \emph{linear spanners},  generalizes \emph{additive spanners} and \emph{multiplicative spanners}.
Recently,  Fomin,  Golovach,  Lochet,  Misra,  Saurabh,  and Sharma initiated the study of additive and multiplicative spanners for directed graphs (IPEC $2020$).
In this article,  we continue this line of research and prove that \textsc{Directed Linear Spanner} parameterized by the number of vertices $n$ admits no polynomial compression of size $\calO(n^{2 - \epsilon})$ for any $\epsilon > 0$ unless $\NP \subseteq \coNP/poly$.
We show that similar results hold for \textsc{Directed Additive Spanner} and \textsc{Directed Multiplicative Spanner} problems.
This sparsification lower bound holds even when the input is a directed acyclic graph and $\alpha,  \beta$ are \emph{any} computable functions of the distance being approximated.
\end{abstract}

%%Graphical abstract
%\begin{graphicalabstract}
%\includegraphics{grabs}
%\end{graphicalabstract}

%%Research highlights
%\begin{highlights}
%\item \textsc{Directed Linear Spanner} parameterized by the number of vertices $n$ admits no polynomial compression of size $\calO(n^{2 - \epsilon})$ for any $\epsilon > 0$ unless $\NP \subseteq \coNP/poly$.
%\item Similar results hold for \textsc{Directed Additive Spanner} and \textsc{Directed Multiplicative Spanner} problems.
%\item The sparsification lower bound holds even when input is a directed acyclic graph and the error terms $\alpha, \beta$ are \emph{any} computable functions of the distance being approximated.
%\end{highlights}

\begin{keyword}
Additive Spanners \sep Multiplicative Spanners  \sep Sparsification  \sep Directed Graphs
%% keywords here, in the form: keyword \sep keyword

%% PACS codes here, in the form: \PACS code \sep code

%% MSC codes here, in the form: \MSC code \sep code
%% or \MSC[2008] code \sep code (2000 is the default)

\end{keyword}

\end{frontmatter}

%% \linenumbers

%% main text
\section{Introduction}
For a graph or a digraph, a spanner is its subgraph that preserves lengths of shortest paths between any two pair of vertices in it up to some additive and/or multiplicative error.
Spanners can be classified as additive spanners, multiplicative spanners, or linear or mixed spanners depending on the type of error is allowed.
We refer readers to the recent survey by Ahmed et al.~\cite{ahmed2020graph} for motivations, applications, and literature regarding this topic.

%A subgraph $H$ of $G$ is its \emph{spanning subgraph} if $V(H) = V(G)$.
%We use $\dist(u, v, G)$ to denote the distance between $u$ and $v$ in $G$.
%For $\beta \ge 0$, a spanning subgraph $H$ of $G$ is said to be an \emph{additive $\beta$-spanner} if $\dist(u, v, H) \le \dist(u, v, G) + \beta$ holds for any pair of vertices $u$ and $v$.
%For $\alpha \ge 1$, a spanning subgraph $H$ of $G$ is said to be a \emph{multiplicative $\alpha$-spanner} if $\dist(u, v, H) \le \alpha \cdot \dist(u, v, G)$ holds for any pair of vertices $u$ and $v$.
%A common generalization of these two concepts is called \emph{linear} or \emph{mixed spanner} which is defined as:
%For $\alpha \ge 1$, $\beta \ge 0$, and a graph $G$, a spanning subgraph $H$ of $G$ is said to be an $(\alpha, \beta)$-spanner if $\dist(u, v, H) \le \alpha \cdot \dist(u, v, G) + \beta$ holds for any pair of vertices $u$ and $v$.
A subgraph $H$ of $G$ is its \emph{spanning subgraph} if $V(H) = V(G)$.
We use $\dist(u, v, G)$ to denote the shortest distance between $u$ and $v$ in $G$.
For $\alpha \ge 1$, $\beta \ge 0$, a spanning subgraph $H$ of $G$ is said to be \emph{$(\alpha, \beta)$-spanner} if $\dist(u, v, H) \le \alpha \cdot \dist(u, v, G) + \beta$ holds for any pair of vertices $u$ and $v$.
These type of spanners are called \emph{linear spanners}.
\emph{Additive spanners} and \emph{multiplicative spanners} are $(1, \beta)$-spanners and $(\alpha, 0)$-spanners, respectively, for some $\beta \ge 1$ and $\alpha > 1$.
Thorup and Zwick~\cite{thorup2006spanners} considered spanners with additive error terms that are sub-linear in the distance being approximated.
They constructed a spanning subgraph $H$ such that $\dist(u, v, H) \le \dist(u, v, G) + f_{\beta}(\dist(u, v, G))$ where $f_{\beta}$ is a sub-linear function of the form $f_{\beta}(d) = c \cdot d^{1/(1 - q)}$ for some constants $c$ and $q \ge 2$.

As any graph is a spanner for itself, a non-trivial question is to find a spanner with as few edges as possible.
Liestman and Shermer \cite{liestman1993additive} proved that for every fixed $\beta \ge 1$, given a graph $G$ and a positive integer $k$, it is \textsf{NP-Complete} to decide whether $G$ admits a $(1, \beta)$-spanner (also called \emph{additive $\beta$-spanner}) with at most $|E(G)| - k$ edges.
It is also known that deciding whether $G$ has an $(\alpha, 0)$-spanner (also called \emph{multiplicative $\alpha$-spanner}) with at most $|E(G)| - k$ edges is \textsf{NP-Complete} for every fixed $\alpha \ge 2$ (\cite{cai1994np}, \cite{peleg1989graph}).

% finding a spa
%A spanning subgraph $H$ of $G$ is said to be \emph{$(f^{nr}_{\alpha}, f^{nr}_{\beta})$-spanner} if 

%To unify these different types of allowed error, we generalize the notation of $(\alpha, \beta)$-spanners to $(f^d_{\alpha}, f^d_{\beta})$-spanners for two functions $f^d_{\alpha}, f^d_{\beta}: \mathbb{N} \mapsto \mathbb{N}$.
%A spanning subgraph $H$ of $G$ is said to be \emph{$(f^d_{\alpha}, f^d_{\beta})$-spanner} if $\dist(u, v, H) \le f^d_{\alpha}(\dist(u, v, G)) \cdot \dist(u, v, G) + f^d_{\beta}(\dist(u, v, G))$ holds for any pair of vertices $u$ and $v$.

Recently, the problem of finding optimum additive, multiplicative, and linear spanners for undirected graph have been studied from the Parameterized Complexity framework.
In this framework,  we measure the computational complexity as a function of the input size of a problem and a secondary measure.
%One of the central notation in this field is that of \emph{kernelization}.
We define relevant notations in Section~\ref{sec:prelims}.
%Informally, we say a problem admits a kernel of size $h(k)$, if given an instance with additional parameter $k$, one can obtain an equivalent instance of size at most $h(k)$ in polynomial time. 
Kobayashi proved that \textsc{Multiplicative Spanner} admits a polynomial kernel of size $\calO(k^2 \alpha^2)$ \cite{kobayashi2018np}, and \textsc{Additive Spanner}, \textsc{Linear Spanner} problems are fixed parameter tractable \cite{kobayashi2020fptalgo}.
Fomin et al. \cite{fomin2020parameterized} started the study of additive and multiplicative spanners for directed graphs.
They proved that \textsc{Directed Multiplicative Spanner} admits a kernel of size $\calO(k^4\alpha^5)$ whereas \textsc{Directed Additive Spanner} is not fixed parameter tractable under a widely believed conjecture.

We continue this line of research and present a sparsification lower bound for these problems.
We remark that the problem of finding a spanner (not necessarily an optimum spanner) becomes `easier' as the error factor increases (see Related Work below).
To present a lower bound for a general case, we consider the spanners that allow the error as \emph{any} computable function of the distance being approximated.
Note that such a formulation is considered by Thorup and Zwick~\cite{thorup2006spanners} for undirected graphs.
We highlight that unlike in their case, where only sub-linear functions are allowed, the following formulation allows any computable function. 

%More precisely, we consider the following problem.
%More precisely, we consider the problem that takes as input a graph $G$, two functions $f^{nr}_{\alpha}, f^{nr}_{\beta}: \mathbb{N} \mapsto \mathbb{N}$, and asks whether there exists a spanning subgraph $H$ of $G$ such that $\dist(u, v, H) \le f^{nr}_{\alpha}(|V(G)|) \cdot \dist(u, v, G) + f^{nr}_{\beta}(|V(G)|)$ holds for any pair of vertices $u$ and $v$.
% and present a sparcification lower bounds for the consider the following generalization of the problem.

\defproblem{\textsc{Directed Linear Spanner}}{Digraph $D$, two monotonically non-decreasing computable functions $f_{\alpha}, f_{\beta}: \mathbb{N} \rightarrow \mathbb{R}_{\ge 0}$, and a positive integer $k$.}{Does there exist a spanning subgraph $H$ of $D$ with at most $|A(D)| - k$ arcs such that $\dist(u, v, H) \le f_{\alpha}(\dist(u, v, D)) \cdot \dist(u, v, D) + f_{\beta}(\dist(u, v, D))$ holds for any pair of vertices $u$ and $v$?}

In this article, we investigate the problem from the perspective of polynomial-time \emph{sparsification}: the method of reducing an input instance to an equivalent object that needs fewer bits to encode.
For example, if input instance comprises a graph or a CNF-formula then the aim is to find an equivalent graph or CNF-formula in which ratio of edges to vertices or clauses to variable is smaller than the original instance.
Consider an instance $(D, f_{\alpha}, f_{\beta}, k)$ of \textsc{Directed Linear Spanner} problem.
We assume, throughout the article, that any function given as a part of input can be encoded using the constant number of bits. 
Hence, this instance can be encoded with $\calO(|V(D)|^2)$ bits as for any non-trivial instance $k \le |A(D)|$. 
The goal of sparsification is to examine existence of a polynomial time algorithm that given an instance $(D, f_{\alpha}, f_{\beta}, k)$, maps it to an equivalent instance of \emph{any} problem that uses $\calO(|V(D)|^{2-\epsilon})$ bits for some $\epsilon > 0$.
We answer this question in the negative.

\begin{theorem}
\label{thm:dir-linear-spanner-sparsification}
Consider two monotonically non-decreasing computable functions $f_{\alpha}, f_{\beta}: \mathbb{N} \rightarrow \mathbb{R}_{\ge 0}$ such that $1 \le f_{\alpha}(1)$ and $2 \le f_{\alpha}(1) + f_{\beta}(1)$.
Unless \emph{$\NP \subseteq \coNP/poly$}, an arbitrary instance $(D, f_{\alpha}, f_{\beta}, k)$ of {\sc Directed Linear Spanner} does not admit a polynomial compression of size $\calO(|V(D)|^{2 - \epsilon})$ for any $\epsilon > 0$, even when $D$ a directed acyclic graph.
\end{theorem}

We justify the condition on the sum of the values of two functions at $1$ in Section~\ref{sec:domset-to-dir-add-spanner}.
%We say functions $f_{\alpha}, f_{\beta}: \mathbb{N} \mapsto \mathbb{N}$ are \emph{sub-polynomial} if for every $\delta \in (0, 1)$, there exists $n_0$ such that $\forall\ e n_0$, $f_{\alpha}(n) + f_{\beta}(n) \le n^{\delta}$.
%In this case, Theorem~\ref{thm:dir-linear-spanner-sparsification} implies that $(D, f_{\alpha}, f_{\beta}, k)$ does not admits a polynomial compression of size $\calO(|V(D)|^{2 - \epsilon})$ for any $\epsilon > 0$ unless $\NP \subseteq \coNP/poly$.
Consider special case when $f_{\alpha}(d) = {\alpha}$ and $f_{\beta}(d) = {\beta}$ for all $d \in \mathbb{N}$ for some non-negative constants $\alpha, \beta$ such that $\alpha + \beta \ge 2$.
%As $n$ is a part of the input, and hence not a constant, in this case there exists $\delta \in (0, 1)$ such that $\alpha + \beta \le n^{\delta}$.
For $\alpha = 1$ and $\beta = 0$, Theorem~\ref{thm:dir-linear-spanner-sparsification} implies, respectively, that \textsc{Directed Additive Spanner} and \textsc{Directed Multiplicative Spanner} do not admit polynomial compressions of size $\calO(|V(D)|^{2 - \epsilon})$ for any $\epsilon > 0$ unless $\NP \subseteq \coNP/poly$. 

%Consider a variation of the problem that is close to the one considered by Thorup and Zwick~\cite{thorup2006spanners}, i.e. the error terms are functions of the distance being approximated.
%Formally, in this case, the objective is to determine whether there exists a spanning subgraph $H$ of $G$ with at least $|A(D)| - k$ arcs such that $\dist(u, v, H) \le f_{\alpha}(\dist(u, v, G)) \cdot \dist(u, v, G) + f_{\beta}(\dist(u, v, G))$ holds for any pair of vertices $u$ and $v$.
% and present a sparcification lower bounds for the consider the following generalization of the problem.
%Small modifications in the proof of Theorem~\ref{thm:dir-linear-spanner-sparsification} imply that for \emph{any} function $f_{\alpha}, f_{\beta}: \mathbb{N} \mapsto \mathbb{N}$ an arbitrary instance $(D, f_{\alpha}, f_{\beta}, k)$ of this variation of the problem does not admit a polynomial compression of size $\calO(|V(D)|^{2 - \epsilon})$ for any $\epsilon > 0$ unless $\NP \subseteq \coNP/poly$.

We remark that Theorem~\ref{thm:dir-linear-spanner-sparsification} \emph{does not} imply that there are digraphs which do not have the spanners that satisfy the properties mentioned in the definition of the problem.
Rather, it implies that to find an optimum linear spanner, one needs information about almost all the arcs in the digraph.
And hence, it is not possible to compress the instance in polynomial time and with non-trivial number of bits without solving it. 
Considering theoretical and practical applications of spanners, we believe Theorem~\ref{thm:dir-linear-spanner-sparsification} provides a non-trivial lower bound.

\paragraph*{Related Work}
It is known that all undirected graphs have additive $2$-spanners with $\calO(n^{1.5})$ edges (\cite{elkin20041plusepsilon}, \cite{aingworth1999fast}, \cite{knudsen2014additive}), additive $4$-spanners with $\calO(n^{1.4})$ edges (\cite{chechik2013new}, \cite{bodwin2020some}), and additive $6$-spanners with $\calO(n^{1.33})$ edges (\cite{baswana2010additive}, \cite{woodruff2010additive}).
Abboud and Bodwin~\cite{abboud20174} proved that for $0 < \epsilon < 1/3$, one cannot compress an input graph into $\calO(n^{1+\epsilon})$ bits, so that one can recover distance information for each pair of vertices within $n^{o(1)}$ additive error. 
A well-known trade-off between the sparsity and the multiplicative factor is:
for any positive integer $\alpha$ and any graph $G$, there is a multiplicative $(2\alpha - 1)$-spanner with $\calO(n^{1 + 1/\alpha})$ edges~\cite{althofer1993sparse}.
This bound is conjectured to be tight based on the popular Girth Conjecture of Erd\H{o}s~\cite{erdos1964theory}.

The general quest for sparsification algorithms is motivated by the fact that they allow instances to be stored, manipulated, and solved more efficiently.
As sparsification preserves the exact answer to the problem, it suffices to solve the sparsified instance.
The notion is fruitful in theoretical \cite{impagliazzo2001problems} and practical \cite{eppstein1997sparsification} settings.
The growing list of problems for which the existence of non-trivial sparsification algorithms has been ruled out under the same assumption includes \textsc{Vertex Cover} \cite{dell2014satisfiability}, \textsc{Dominating Set} \cite{jansen2017sparsification}, \textsc{Feedback Arc Set} \cite{jansen2017sparsification}, \textsc{Treewidth} \cite{jansen2015sparsification}, \textsc{List $H$-Coloring} \cite{chen2020sparsification}, and \textsc{Boolean Constraint Satisfaction} problems \cite{chen2020best}.

\paragraph*{Organization of the article}
We organize the remaining article as follows.
In Section~\ref{sec:prelims}, we present some preliminaries.
In Section~\ref{sec:domset-to-dir-add-spanner},  we present a parameter preserving reduction from \textsc{Dominating Set} to \textsc{Directed Linear Spanner}.
We use this reduction to present a proof of Theorem~\ref{thm:dir-linear-spanner-sparsification}.
%We also mention the modifications in the proof to prove similar results for other problems.
We conclude this article in Section~\ref{sec:conclusion}.
%Due to space constraints, we move proofs of statements marked with $(\star)$ to the appendix.  
\section{Preliminaries}
\label{sec:prelims}

We denote the set of positive integers and the set of non-negative real numbers by $\mathbb{N}$ and $\mathbb{R}_{\ge 0}$, respectively.
For a positive integer $q$,  we denote set $\{1, 2,\dots,  q\}$ by $[q]$.

We consider graphs and directed graphs with a finite number of vertices that do not have loops or multiple edges/arcs as they are irrelevant for distances.
For an undirected graph $G$,  by $V(G)$ and $E(G)$ we denote the set of vertices and edges of $G$ respectively.
Two vertices $u, v$ are said to be \emph{adjacent} in $G$ if there is an edge $(u, v) \in E(G)$.
The \emph{neighborhood} of a vertex $v$, denoted by $N_G(v)$, is the set of vertices adjacent to $v$.
The \emph{closed neighborhood} of a vertex is $N_G[v] = N_G(v)\cup \{v\}$.
%We omit the subscripts in the notation for neighbourhood if the context is clear.
We say $u \in V(G)$ \emph{dominates} $v \in V(G)$ if $v$ is in $N[u]$.
A set $X \subseteq V(G)$ is a \emph{dominating set} of $G$ if $V(G) = N_G[X]$.
%In other words, for every vertex $u$ in $V(G)$, there is a vertex in $X$ that dominates $u$.
For a {\em directed graph (or digraph)} $D$,  by $V(D)$ and $A(D)$ we denote the sets of vertices and directed arcs in $D$,  respectively.
%Vertices $u, v$ are said to be start-point and end-point of arc $(u, v)$.
%Also, $u$ is called in-neighbour of $v$ and $v$ is called out-neighbour of $u$.
%For a subset $S \subseteq V(D)$, by $D - S$ and $D[S]$ we denote the graph obtained by deleting vertices in $S$ from $D$ and the graph obtained by removing vertices in $V(D) \setminus S$ from $D$, respectively.
For $F \subseteq A(D)$, $D - F$ is the graph obtained by deleting arcs in $F$ from $D$.
%For $X, Y \subseteq V(D)$, $A(X, Y)$ denotes the set of arcs with start-point in $X$ and end-point in $Y$.
%For $F \subseteq A(D)$, $V(F)$ denotes the set of endpoints of arcs in $F$.

A \emph{directed path} $P$ in $D$ is an ordered sequence of vertices $\langle v_1,  v_2, \dots,  v_q\rangle$ such that $(v_i,  v_{i + 1}) \in A(D)$ for all $i \in [q - 1]$.
We denote the set of arcs $\{(v_i,  v_{i + 1})\ |\ i \in [q - 1]\}$ by $A(P)$.
The length of directed path $P$ is $|A(P)|$.
For two vertices $u, v \in V(D)$,  $\dist(u, v, D)$ denotes the length of a shortest directed path from $u$ to $v$.
If there is no directed path from $u$ to $v$, then we assign $\dist(u, v, D) = +\infty$.
Note that $\dist(u, v, D)$ may not be equal to $\dist(v, u, D)$.
Consider two directed paths $P_1 = \langle v_1,  v_2,  \dots,  v_q \rangle$ and $P_2 = \langle u_1,  u_2,  \dots,  u_p \rangle$.
If $v_q = u_1$ then we denote directed path $P = \langle v_1,  \dots, v_q = u_1, \dots u_p \rangle$ by $P_1 \circ P_2$.
A digraph $D$ is said to be \emph{acyclic} if the following statement is true for any two distinct vertices: 
if there is a directed path from $u$ to $v$ then there is no directed path from $v$ to $u$. 

A {\em subdivision} of an arc $(u, v) \in A(D)$ is an operation that deletes arc $(u, v)$, adds a vertex $w$ to $V(D)$, and adds arcs $(u, w)$ and $(w, v)$.
We say arc $(u, v)$ is subdivided $q$ times if we delete arc $(u, v)$ and add a directed path of length $(q + 1)$ from $u$ to $v$. 
%With slight abuse of notation, we say arc is subdivided zero times when we do not change digraph. 
If digraph $D’$ is obtained from $D$ by subdividing $q$ times arc $(u,  v)$ then $\dist(u,  v,  D’) = q + 1$. 
A \emph{contraction} of an arc $(u, v) \in A(D)$ is an operation that results in a digraph $D'$ on the vertex set $V(D') = (V(D) \setminus \{u,v\}) \cup \{w\}$ with
$A(D')=\{(x, y) \mid (x, y) \in A(D) \textnormal{ and } x,y \in V(D')\} \cup
\{(x, w) \mid (x, u) \in A(D)\} \cup \{(w, y) \mid (u, y) \in A(D) \} \cup \{ (x, w) \mid
(x, v) \in A(D)\} \cup \{(w, y) \mid (v, y) \in A(D)\}$.

%Consider two problems $\calQ, \calR \subseteq \Sigma$.
%We say instance $x$ of $\calQ$ is \emph{equivalent} with instance $y$ of $\calR$ if the following condition holds: $x$ is a \yes\ instance of $Q$ if and only if $x'$ is a \yes\ instance of $\calQ'$.
We refer the readers to the recent books \cite{cygan2015parameterized}, \cite{fomin2019kernelization} for the detailed introduction of Parameterized Complexity theory.
A parameterized language $\calQ$ is a subset of $\Sigma^{*} \times \mathbb{N}$, where $\Sigma$ is a finite alphabet.
The second component of a tuple $(x, k) \in \Sigma^{*} \times \mathbb{N}$ is called the \emph{parameter}.
A parameterized language is said to be fixed-parameter tractable (or \FPT) if there exists an algorithm that given a tuple $(x, k) \in \Sigma^{*} \times \mathbb{N}$, runs in $f(k)\cdot |x|^{\calO(1)}$, for some computable function $f(\cdot)$, and correctly determines whether $(x, k) \in \calQ$.
The notion of \emph{kernelization} is used to capture various forms of efficient preprocessing.
We define it in its general form.
%\begin{definition}[Polynomial Compression~\cite{bodlaender2014kernelization}]
%Let $\calQ, \calQ’ \subseteq \Sigma \times \mathbb{N}$ be parameterized problems and let $h: \mathbb{N} \rightarrow \mathbb{N}$ be a computable function.
%A \emph{generalized kernel} for $\calQ$ into $\calQ’$ of size $h(k)$ is an algorithm that, on input $(x, k) \in \Sigma \times \mathbb{N}$, takes time polynomial in $|x| + k$ and outputs an instance $(x’, k’)$ such that:
%$(i)$ $|x’|$ and $k’$ are bounded by $h(k)$, and
%$(ii)$ $(x’, k’) \in \calQ’$ if and only if $(x, k) \in Q$
%\begin{enumerate}
%\item $|x’|$ and $k’$ are bounded by $h(k)$, and
%\item $(x’, k’) \in \calQ’$ if and only if $(x, k) \in Q$. 
%\end{enumerate}
%\end{definition}
%The algorithm is a kernel for $\calQ$ if $\calQ’ = \calQ$.
%It is a polynomial (generalized) kernel if $h(k)$ is a polynomial.
%We use the notation of generalized kernel to state our result.
\begin{definition}[Definition $1.5$ in \cite{fomin2019kernelization}] A \emph{polynomial compression} of a parameterized language $\calQ \subseteq \Sigma^{*} \times \mathbb{N}$ into a language $\calR \subseteq \Sigma^{*}$ is an algorithm that takes as input an instance $(x, k) \in \Sigma^{*} \times \mathbb{N}$, runs in time polynomial in $|x| + k$, and returns a string $y$ such that:
$(i)$ $|y| \le p(k)$ for some polynomial $p(\cdot)$, and
$(ii)$ $y \in \calR$ if and only if $(x, k) \in \calQ$.
\end{definition}
We need the following result regarding sparsification.
\begin{proposition}[Theorem~$4$ in \cite{jansen2017sparsification}]
\label{prop:dom-set-sparsification}
Unless \emph{$\NP \subseteq \coNP/poly$}, \textsc{Dominating Set} parameterized by the number of vertices $n$ does not admit a polynomial compression of size $\calO(n^{2 - \epsilon})$ for any $\epsilon > 0$.
\end{proposition}

\section{Proof of Theorem~\ref{thm:dir-linear-spanner-sparsification}}
\label{sec:domset-to-dir-add-spanner}
To prove the theorem, we present a reduction from \textsc{Dominating Set} to \textsc{Directed Linear Spanner}.
In the \textsc{Dominating Set} problem, an input is an undirected graph $G$ and an integer $l$.
The objective is to decide whether there is a dominating set of size at most $l$ in $G$. 

Consider an instance $(D, f_{\alpha}, f_{\beta}, k)$ of \textsc{Directed Linear Spanner} such that $f_{\alpha}(1) + f_{\beta}(1) < 2$.
Recall that $k$ is a positive integer.
By the definition of the problem, it is safe to consider that $D$ does not have parallel arcs.
Assume that there exist a set of arcs $F \subseteq A(D)$ of size at least $k$ such that $D - F$ satisfies the properties mention in the problem statement.
For the endpoints of an arc $(u, v) \in F$, we have $\dist(u, v, D - F) \le 2 \le f_{\alpha}(\dist(u, v, D )) + f_{\beta}(\dist(u, v, D)) \le f_{\alpha}(1) + f_{\beta}(1) < 2$, which is a contradiction.
Hence, our assumption is wrong, and no such set of arcs exists.
In this case the input is a \no-instance.
This fact can be encoded in the constant bit-size.
To avoid this trivial case, we consider the functions for which $2 \le f_{\alpha}(1) + f_{\beta}(1)$.
As a technical requirement, we need $1 \le f_{\alpha}(1)$ and $f_{\alpha}, f_{\beta}$ are monotonically non-decreasing functions.

%We present a reduction and proof of correctness for the instance  $(D, f_{\alpha}, f_{\beta}, k)$  of \textsc{Directed Linear Spanner} i.e. when the error is allowed as a function of number of vertices.
%We later mention how to modify the reduction for the other case.

\paragraph*{Reduction}
The reduction takes as input an instance $(G, l)$ of \textsc{Dominating Set} and two monotonically non-decreasing computable functions $f_{\alpha}, f_{\beta}: \mathbb{N} \rightarrow \mathbb{R}_{\ge 0}$.
It outputs an instance $(D, f_{\alpha}, f_{\beta}, k)$ of \textsc{Directed Linear Spanner}.
Let $V(G) = \{v_1, v_2, \dots, v_{|V(G)|}\}$.
The reduction creates digraph $D$ as follows:
\begin{itemize}
\item[-] It creates the following four copies of $V(G)$: $R_{l}, R_{c}, R_{r}$, and $B$. 
For every vertex $v_i$ in $V(G)$, let $R_{l}[i], R_{c}[i], R_{r}[i], B[i]$ denote these four copies.
It adds a new vertex $w$.
\item[-] For every vertex $v_i$ in $V(G)$, it adds the following six arcs $(w, R_l[i])$, $(w, R_c[i])$, $(w, B[i])$, $(R_l[i], R_c[i])$, $(R_c[i], R_r[i])$, and $(R_r[i], B[i])$.
\item[-] For every edge $(v_i, v_j)$ in $E(G)$, it adds the following arcs: $(R_r[i], B[j])$, $(R_r[j], B[i])$.
\item[-] Let $t = f_{\alpha}(1) + f_{\beta}(1)$.
Recall that, by our assumption on the functions, $t \ge 2$.
If $\lfloor t \rfloor = 2$, then the reduction contracts arc $(R_c[i], R_r[i])$ for every $i \in [|V(G)|]$.
If $\lfloor t \rfloor = 3$, it does not modify the graph.
If $\lfloor t \rfloor \ge 4$, it subdivides $\lfloor t - 3 \rfloor$-times arc $(R_c[i], R_r[i])$ for every $i \in [|V(G)|]$.
\end{itemize}

This completes the construction of $D$.
The reduction sets $k = 2 \cdot |V(G)| - l$ and returns the instance $(D, f_{\alpha}, f_{\beta}, k)$.
See Figure~\ref{fig:reduction-illustration} for an illustration.

\begin{figure}[t]
\begin{center}
\includegraphics[scale=0.70]{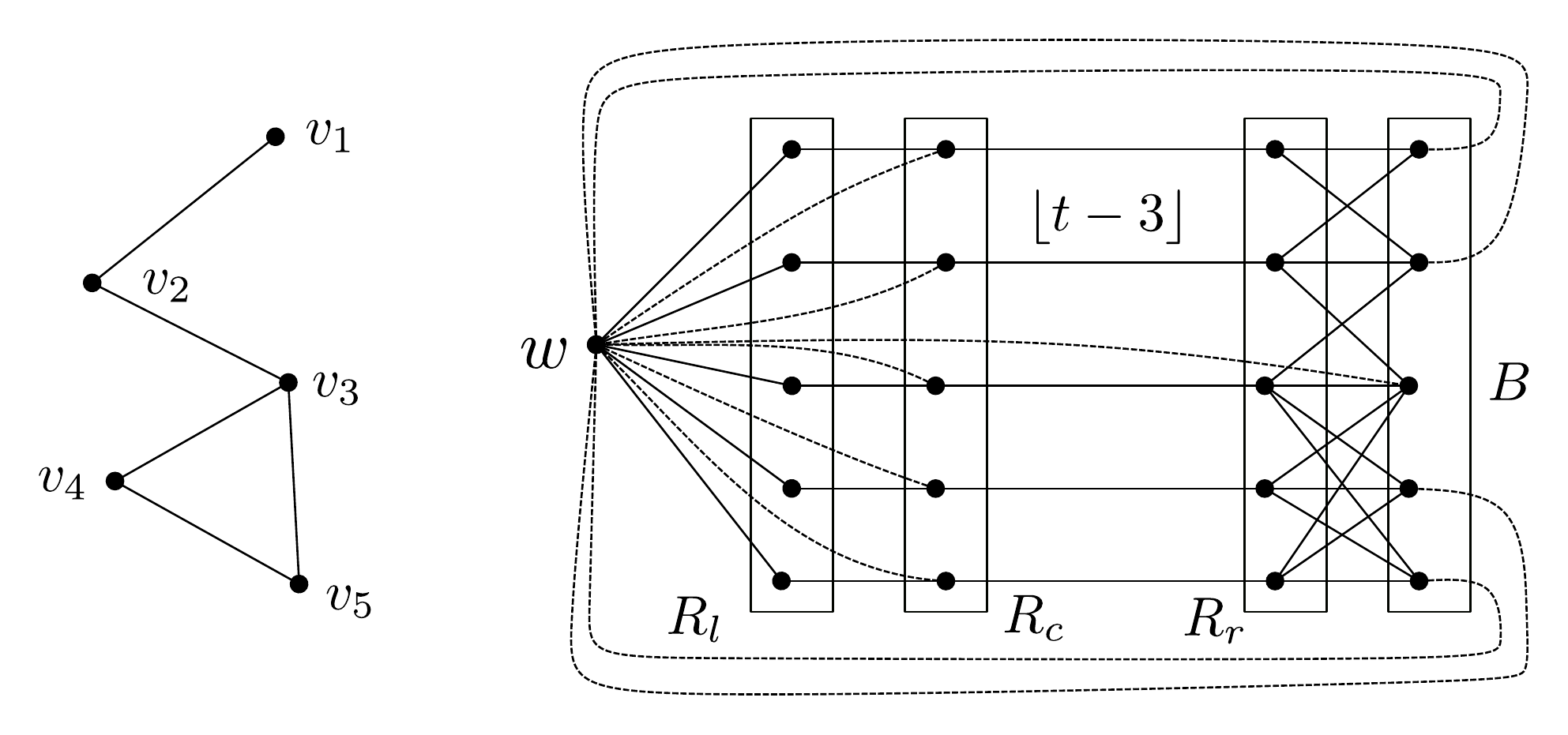}
\end{center}
\caption{Overview of the reduction. 
Left side is an input graph $G$ and right side shows the corresponding digraph $D$ constructed by the reduction.
All arcs in $D$ are directed from their left end-point to right end-point.
Every arc starting in $R_c$ and ending in $R_r$ is sub-divided $\lfloor t - 3 \rfloor$ times.
For every vertex in $V(G)$, four vertices on the same horizontal line are the vertices corresponding to it in $V(D)$.   
Only the dotted arcs can be omitted from a spanner of $D$. 
\label{fig:reduction-illustration}}
\end{figure}

\begin{lemma}
\label{lemma:forward-direction}
If $(G, l)$ is a \yes\ instance of \textsc{Dominating Set} then $(D, f_{\alpha}, f_{\beta}, k)$ is a \yes\ instance of \textsc{Directed Linear Spanner}.
\end{lemma}
\begin{proof}
Let $X$ be a dominating set of size at most $l$ in $G$.
Consider the subset of arcs $F \subseteq A(D)$ that contains all the arcs of the form $(w, B[i])$ and arcs $(w, R_c[i])$ corresponding to every vertex \emph{not} in $X$.
Formally, 
$$F := \{(w, B[i]) \ |\ \forall\ i \in [|V(G)|]\}\ \cup\ \{(w[i], R_c[i])\ |\ \forall i \in [|V(G)|]\  s.t.\ v_i \not\in X\}.$$
Note that $|F| \ge |V(G)| + |V(G)| - l \ge k$.
We argue that $D - F$ satisfies the properties mentioned in the definition of \textsc{Directed Linear Spanner}.

It is easy to see that for any vertex $u \in V(D) \setminus \{w\}$ and vertex $v \in V(D) \setminus B$, if there is a directed path from $u$ to $v$ in $D$ then the same path is also present in $D - F$.
Hence, for all pairs of such vertices, we have $\dist(u, v, D - F) = \dist(u, v, D)$.

Consider the case when $u = w$ and $v \in V (D)\setminus B$.
It is easy to verify that if $\dist(u, v, D) = d$, then $\dist(u, v, D - F) \le d + 1$.
Hence, it is sufficient to argue that $d + 1 \le  f_{\alpha}(d) \cdot d + f_{\beta}(d)$ for every $d \ge 1$.
Recall that $1 \le f_{\alpha}(1)$ and $2 \le f_{\alpha}(1) + f_{\beta}(1)$.
For $d \ge 1$, this implies that $(d - 1) + 2 \le (d - 1) \cdot f_{\alpha}(1) + f_{\alpha}(1) + f_{\beta}(1) \le d \cdot f_{\alpha}(1) + f_{\beta}(1)$.
As $f_{\alpha}, f_{\beta}$ are monotonically non-decreasing functions and $d \ge 1$, we have
$d + 1 \le d \cdot f_{\alpha}(d) + f_{\beta}(d)$.
This implies that $\dist(u, v, D - F) \le f_{\alpha}(\dist(u, v, D)) \cdot \dist(u, v, D) + f_{\beta}(\dist(u, v, D))$.

For any $i \in [|V(G)|]$ there are following three directed paths that starts at $w$ and ends at $B[i]$:
$(i)$ $\langle w, B[i]\rangle$, 
$(ii)$ $\langle w, R_{c}[j]\rangle \circ \langle R_{c}[j], \dots, R_{r}[j]\rangle \circ \langle R_{r}[j], B[i]\rangle$, and
$(iii)$ $\langle w, R_{l}[j], R_{c}[j]\rangle \circ \langle R_{c}[j], \dots, R_{r}[j]\rangle \circ \langle R_{r}[j], B[i]\rangle$.
For the last two types of directed paths, $j \in [|V(G)|]$ such that either $i = j$ or $(v_i, v_j) \in E(G)$.
The lengths of these three types of directed paths are $1$, $2 + \lfloor t - 2 \rfloor$, and $3 + \lfloor t  - 2 \rfloor$, respectively.
As $X$ is a dominating set, for any $i \in [|V(G)|]$, at least one directed path of the second type is present in $D - F$.
This implies $\dist(w, B[i], D - F) =  2 + \lfloor t - 2 \rfloor \le t = f_{\alpha}(1) + f_{\beta}(1) = f_{\alpha}(\dist(w, B[i], D)) \cdot \dist(w, B[i], D) + f_{\beta}(\dist(w, B[i], D))$.

This implies $D - F$ satisfies the properties mentioned in the problem definition of \textsc{Directed Linear Spanner}. 
Hence, if $(G, \ell)$ is a \yes\ instance of \textsc{Dominating Set} then $(D, f_{\alpha}, f_{\beta}, k)$ is a \yes\ instance of \textsc{Directed Linear Spanner}.
\end{proof}

\begin{lemma}
\label{lemma:backward-direction}
If $(D, f_{\alpha}, f_{\beta}, k)$ is a \yes\ instance of \textsc{Directed Linear Spanner} then $(G, l)$ is a \yes\ instance of \textsc{Dominating Set}.
\end{lemma}
\begin{proof}
Let $F \subseteq A(D)$ be the set of arcs in $A(D)$ such that $D - F$ satisfies the properties mentioned in the definition of \textsc{Directed Linear Spanner}.
We argue that arcs in $F$ are of the form $(w, R_c[i])$ or $(w, B[i])$ for some $i \in [|V(D)|]$.
For any arc $(u, v)$ in $A(D)$ that is not of the above forms, there is no directed path from $u$ to $v$ in digraph $D - (u, v)$.
Hence, if $(u, v)$ is in $F$, then $\dist(u, v, D - F) = \infty$ which is not upper bounded by $f_{\alpha}(\dist(u, v, D)) \cdot \dist(u, v, D) + f_{\beta}(\dist(u, v, D))$.
This contradicts the fact that $G - F$ satisfies the properties mentioned in the definition of the problem.

We say arc $(w, R_c[i])$ corresponds to vertex $v_i$ in $V(G)$ for every $i \in [|V(G)|]$.
Consider the subset $X$ of vertices in $V(G)$ whose corresponding to arcs are \emph{not} in $F$.
Formally, $X := \{v_i\ |\ i \in [|V(G)|] \textrm{ and } (w, R_c[i]) \not\in F\}$.
We argue that set $X \cup (V(G) \setminus N[X])$ is a dominating set of size at most $l$ in $G$.

We claim that for a vertex $v_i \in V(G) \setminus N_G[X]$, set $F$ does not contain arc $(w, B[i])$.
Assume, for the sake contradiction, that $(w, B[i])$ is in $F$.
By the construction of $X$, arc $(w, R_c[i])$ is present in $F$.
Any directed path from $w$ to $B[i]$ in $D - F$ is either of the form $\langle w, R_{c}[j]\rangle \circ \langle R_{c}[j], \dots, R_{r}[j]\rangle \circ \langle R_r[j], B[i]\rangle$ or $\langle w, R_{l}[j], R_{c}[j]\rangle \circ \langle R_{c}[j], \dots R_{r}[j]\rangle \circ \langle R_r[j], B[i]\rangle$ for some $j \in [|V(G)|]$ such that $(v_i, v_j) \in E(G)$.
As $v_i \in V(G) \setminus N_G[X]$, for every $j \in [|V(G)|]$ such that $(v_i, v_j) \in E(G)$, we have $v_j \not\in X$.
Hence, arc $(w, R_{c}[j])$ is in $F$ and not present in $D - F$.
This implies any directed path from $w$ to $B[i]$ is of the second form.
Recall that the length of such directed path is $3 + \lfloor t - 2\rfloor$.
This implies $3 + \lfloor t - 2\rfloor \le f_{\alpha}(\dist(w, B[i], D)) \cdot \dist(w, B[i], D) + f_{\beta}(\dist(w, B[i], D)) \le f_{\alpha}(1) \cdot 1 + f_{\beta}(1) = t$, which is a contradiction.
Hence, our assumption was wrong and for any $v_i \in V(G) \setminus N_G[X]$, set $F$ does not contain arc $(w, B[i])$.

Define set $Y = V(G) \setminus N_G[X]$.
Note that $X \cup Y$ dominates every vertex in $V(G)$.
It remains to argue the bound on the size of $X \cup Y$.
By the definition of $X$, set $F$ contains $|V(G)| - |X|$ many arcs of the form $(w, R_c[i])$ for some $i \in [|V(G)|]$.
By the definition of $Y$, set $F$ contains $|V(G)| - |Y|$ many arcs of the form $(w, B[i])$ for some $i \in [|V(G)|]$.
As $|F| \ge 2|V(G)| - l$, we have $|X| + |Y| \le l$.
This implies there is a dominating set of size at most $l$ in $G$, and hence $(G, l)$ is a \yes\ instance of \textsc{Dominating Set}. 
This concludes the proof of the lemma.
\end{proof}

\begin{proof}(\emph{of Theorem~\ref{thm:dir-linear-spanner-sparsification}})
Assume, for the sake of contradiction, that there is an algorithm $\calA$ that given a constant $\epsilon > 0$ and an instance $(D_1, f_{\alpha}, f_{\beta}, k)$ of \textsc{Directed Linear Spanner}, where $D_1$ is a directed acyclic graph, runs in polynomial time and computes its polynomial compression of size $\calO(|V(D_1)|^{2 - \epsilon})$.
%Recall that there exists $\delta \in [0, 1)$ and $n_0 \in \mathbb{N}$ that satisfies the following: $\forall\ n \ge n_0$, $f_{\alpha}(d) + f_{\beta}(d) \le n^{\delta}$.

Consider the following algorithm $\calB$ that takes as input an instance $(G, l)$ of \textsc{Dominating Set} and returns its polynomial compression.
%We assume, without loss of generality, that $|V(G)| \ge n_0$.
Algorithm $\calB$ runs the reduction mentioned above as a subroutine with input $(G, l)$ and functions $f_{\alpha}, f_{\beta}$.
Let $(D, f_{\alpha}, f_{\beta}, k)$ be the equivalent instance returned by the reduction.
It is easy to verify that $D$ is a directed acyclic graph.
Algorithm~$\calB$ uses Algorithm~$\calA$ as a subroutine to obtain a polynomial compression of size $\calO(|V(D)|^{2 - \epsilon})$.
It then returns this as polynomial compression for $(D, l)$.
This completes the description of Algorithm~$\calB$. 

Algorithm~$\calB$ runs in polynomial time and its correctness follows from Lemma~\ref{lemma:forward-direction}, Lemma~\ref{lemma:backward-direction}, and the correctness of Algorithm~$\calA$.
%As $|V(G)| \ge n_0$, we have $f_{\alpha}(|V(G)|) + f_{\beta}(|V(G)|) \le |V(G)|^{\delta}$.
By the description of Algorithm~$\calB$, the number of vertices in $D$ is at most $4 \cdot |V(G)| + |V(G)| \cdot (f_{\alpha}(1) + f_{\beta}(1)) + 1 \in \calO(|V(G)|)$. 
This implies Algorithm~$\calB$ computes a polynomial compression of \textsc{Dominating Set} of size $\calO(|V(G)|^{2 - \epsilon})$.
But, this contradicts Proposition~\ref{prop:dom-set-sparsification}.
Hence, our assumption was wrong and no such algorithm exists.
This concludes the proof of the theorem.
\end{proof}

%The reduction and the proof of correctness can be easily adopted for the variation of the problem mentioned in the first section.
%Recall that, in this version the objective is to determine whether there exists a spanning subgraph $H$ of $G$ with at least $|A(D)| - k$ arcs such that $\dist(u, v, H) \le f_{\alpha}(\dist(u, v, G)) \cdot \dist(u, v, G) + f_{\beta}(\dist(u, v, G))$ holds for any pair of vertices $u$ and $v$.
%We assign $t = f_{\alpha}(1) + f_{\beta}(1)$ while constructing digraph $D$.
%This ensures that arcs from $w$ to $B[i]$ that are of form $\langle w, R_{c}[j]\rangle \circ \langle R_{c}[j], \dots, R_{r}[j] \rangle \circ \langle R_r[j], B[i]\rangle$ are of length $t$, where as paths of the form $\langle w, R_{l}[j], R_{c}[j] \rangle \circ \langle R_{c}[j], \dots R_{r}[j]\rangle \circ \langle R_r[j], B[i] \rangle$ are of length $t + 1$.
%Hence, if arc $(w, B[i])$ is deleted to obtain a spanner, then at least one path of the first kind should be present in the spanner.
\section{Conclusion}
\label{sec:conclusion}
In this article, we proved that unless $\NP \subseteq \coNP/poly$, \textsc{Directed Linear Spanner} parameterized by the number of vertices $n$ admits no generalized kernel of size $\calO(n^{2 - \epsilon})$ for any $\epsilon > 0$.
This lower bound holds even when input is a directed acyclic graph and $\alpha, \beta$ are \emph{any} computable functions of the distance being approximated.
%We can prove similar result for \textsc{Directed Additive Spanner} problem and \textsc{Directed Multiplicative Spanner} problem.
Abboud and Bodwin~\cite{abboud20174} proved that unconditional  sparsification lower bound for undirected graphs with weaker constants in the exponent.
It will be interesting to investigate whether their lower bound can be strengthen in case of directed graphs.

We can extend our result to more generalized problem at the cost stronger condition on the error function.
%Consider the following generalization of the problem in which the error function is not necessarily linear in terms of the distance. 
Consider the generalization of the problem, called \textsc{Directed Spanner}, in which the error function $f : \mathbb{N} \rightarrow \mathbb{R}_{\ge 0}$ is not restricted to be linear in terms of the distance.
%In \textsc{Directed Spanner} problem, an input is a digraph $D$, a computable function $f : \mathbb{N} \rightarrow \mathbb{R}_{\ge 0}$, and a positive integer $k$.
%The objective is to decide whether there exists a spanning subgraph $H$ of $D$ with at most $|A(D)| - k$ arcs such that $\dist(u, v, H) \le f(\dist(u, v, D))$ for any pair of vertices in $u$ and $v$.
Using the identical arguments, it is not hard to see that if $d + 1 \le f(d)$ for every $d \ge 1$, then the \textsc{Directed Spanner} problem does not admit a polynomial compression of size $\calO(|V(D)|^{2 - \epsilon})$ for any $\epsilon > 0$, even when $D$ a directed acyclic graph.
Note that this condition, i.e. $d + 1 \le f(d)$ for every $d \ge 1$, is stronger than the conditions, i.e. $1 \le f_{\alpha}(1)$, $2 \le f_{\alpha}(1) + f_{\beta}(1)$, and both $f_{\alpha}, f_{\beta}$ are monotonically non-decreasing, we used to prove the similar result for \textsc{Directed Linear Spanner}.

%% The Appendices part is started with the command \appendix;
%% appendix sections are then done as normal sections
%% \appendix

%% \section{}
%% \label{}

%% If you have bibdatabase file and want bibtex to generate the
%% bibitems, please use
%%
\bibliographystyle{elsarticle-num} 
\bibliography{references.bib}

%% else use the following coding to input the bibitems directly in the
%% TeX file.

%\begin{thebibliography}{00}

%% \bibitem{label}
%% Text of bibliographic item

%\bibitem{}

%\end{thebibliography}
\end{document}